\documentclass[conference]{IEEEtran}
\IEEEoverridecommandlockouts

\usepackage{euscript}
\DeclareMathAlphabet\EuRoman{U}{eur}{m}{n}
\SetMathAlphabet\EuRoman{bold}{U}{eur}{b}{n}
\newcommand{\eurom}{\EuRoman}

\usepackage{bbm}
\usepackage{bm}
\usepackage{amssymb}

\usepackage[utf8]{inputenc} 
\usepackage[T1]{fontenc}
\usepackage{url}
\usepackage{ifthen}
\usepackage[noadjust]{cite}
\usepackage[cmex10]{amsmath}
\usepackage{amsthm}

\theoremstyle{definition}

\newtheoremstyle{mystyle}%                % Name
  {}%                                     % Space above
  {}%                                     % Space below
  {\itshape}%                                     % Body font
  {}%                                     % Indent amount
  {\bfseries}%                            % Theorem head font
  {.}%                                    % Punctuation after theorem head
  { }%                                     \newline
  {}%                                     % Theorem head spec (can be left empty, meaning `normal')

\theoremstyle{mystyle}
\newtheorem{theorem}{Theorem}%[section]
\theoremstyle{mystyle}
\newtheorem{lemma}{Lemma}%[section]
\theoremstyle{mystyle}
\theoremstyle{mystyle}
\theoremstyle{remark}
\newtheorem{remark}{Remark}%[section]

\usepackage[english]{babel}

\usepackage{mleftright}

\usepackage{tikz}

\usetikzlibrary{math}
\usetikzlibrary{shapes}
\usetikzlibrary{positioning}
\usetikzlibrary{arrows}
\usetikzlibrary{arrows.meta}

\usepackage{pgfplots}
\pgfplotsset{compat=newest}
\pgfplotsset{plot coordinates/math parser=false}
\newlength\figureheight
\newlength\figurewidth

\usepackage{xcolor}
\usepackage{hyperref}
\usepackage{soul}
\hypersetup{%
  colorlinks = true,
  linkcolor  = red,
  citecolor  = blue,
  urlcolor   = blue
}
\usepackage{soul}
\setul{1pt}{.4pt}

\usepackage[nameinlink]{cleveref}
\Crefname{figure}{}{}

\usepackage{mathtools}

\newcommand{\abs}[1]{\mleft| #1 \mright|}
\newcommand{\norm}[1]{\mleft\| #1 \mright\|}

\newcommand{\C}{\eurom{C}}

\newcommand{\I}{\eurom{I}}
\newcommand{\mi}[2]{{\I}\mleft(#1 \, ; #2 \mright)}

\newcommand{\h}{\eurom{h}}
\newcommand{\ent}[1]{{\h}\mleft(#1\mright)}

\newcommand{\Exp}{\eurom{E}}
\newcommand{\expect}[1]{{\Exp}\mleft[#1\mright]}

\newcommand{\varz}{\sigma_z^2}

\newcommand{\Yv}{{\bf Y}}
\newcommand{\Hv}{{\textsf{H}}}
\newcommand{\Mv}{{\textsf{M}}}

\newcommand{\Xv}{{\bf X}}

\newcommand{\xv}{{\bf x}}

\newcommand{\Zv}{{\bf Z}}

\newcommand{\Iv}{\textsf{I}}
\newcommand{\Uv}{\textsf{U}}
\newcommand{\Lv}{\mathsf{\Lambda}}

\newcommand{\Dv}{\mathsf{D}}
\newcommand{\Vv}{\textsf{V}}

\newcommand{\sX}{{\cal X}}

\newcommand{\Vol}[2][]{\mathrm{Vol}_{#1} \mleft( #2 \mright) }
\newcommand{\Bocs}[2][]{\mathrm{Box}_{#1} \mleft( #2 \mright) }
\newcommand{\ball}[2][]{{\cal B}_{#1} \lr{#2}}

\newcommand{\lr}[1]{\mleft( #1 \mright)}
\newcommand{\lrs}[1]{\mleft[ #1 \mright]}
\newcommand{\lrc}[1]{\mleft\{ #1 \mright\}}
  
\newcommand{\N}{\eurom{N}}
\newcommand{\K}{\eurom{K}}
\newcommand{\R}{\eurom{R}}

\newcommand{\Pj}{\eurom{P}}

\title{The Capacity of Fading Vector Gaussian Channels Under Amplitude Constraints on Antenna Subsets}
\author{ }
\date{ }

\begin{document}

\author{%
   \IEEEauthorblockN{Antonino Favano\IEEEauthorrefmark{1}\IEEEauthorrefmark{2},
                     Marco Ferrari\IEEEauthorrefmark{2},
                      Maurizio Magarini\IEEEauthorrefmark{1},
                     and Luca Barletta\IEEEauthorrefmark{1}}
   \IEEEauthorblockA{\IEEEauthorrefmark{1}%
                     Politecnico di Milano,
                     Milano, Italy,
                     \{antonino.favano, maurizio.magarini, luca.barletta\}@polimi.it}
   \IEEEauthorblockA{\IEEEauthorrefmark{2}%
                     CNR-IEIIT, Milano, Italy,
                     marco.ferrari@ieiit.cnr.it}
 }

\maketitle

\begin{abstract}
Upper bounds on the capacity of vector Gaussian channels affected by fading are derived under peak amplitude constraints at the input. The focus is on constraint regions that can be decomposed in a Cartesian product of sub-regions. This constraint models a transmitter configuration employing a number of power amplifiers less than or equal to the total number of transmitting antennas. In general, the power amplifiers feed distinct subsets of the transmitting antennas and partition the input in independent subspaces. Two upper bounds are derived: The first one is suitable for high signal-to-noise ratio (SNR) values and, as we prove, it is tight in this regime; The second upper bound is accurate at low SNR. Furthermore, the derived upper bounds are applied to the relevant case of amplitude constraints induced by employing a distinct power amplifier for each transmitting antenna.
\end{abstract}

\section{Introduction}
Amplitude constraints accurately model the main limitation induced by power amplifiers due to their nonlinear behavior. For this reason, the evaluation of the channel capacity under peak amplitude constraints is a research topic of great practical interest. One of the first contributions in this field is thanks to Smith~\cite{Smith}. In his work, he investigates the capacity of scalar Gaussian channels and the capacity-achieving input distribution. He proves that the optimal input distribution is discrete and composed of a finite number of mass points. In~\cite{Shamai}, the authors extend Smith's findings to quadrature Gaussian channel under amplitude constraints on the norm of the input, proving that the capacity-achieving input distribution is again discrete, made of a finite number of mass points, and also uniformly distributed in its phase. A further generalization to vector Gaussian channels is presented in~\cite{Rassouli}.
Other significant results on the discreteness of the optimal input distribution are presented in~\cite{Tchamkerten2004,Chan2005,Mamandipoor2014}.

In~\cite{McKellips}, McKellips presents a tight upper bound on the capacity of scalar Gaussian channels under peak amplitude constraints. The authors of~\cite{thangaraj2017capacity} rederive the McKellips' upper bound through a dual capacity expression and generalize it to higher dimensions. Furthermore, they improve on McKellips' result and define a more accurate upper bound, which they refer to as \emph{refined} upper bound. In~\cite{ourISIT2021}, the present authors define a numerical algorithm to evaluate an arbitrarily precise estimate of the channel capacity and of its capacity-achieving distribution.

In the aforementioned works, the amplitude constraint is set on the norm of the input vector, which  correctly models the limitation induced by a single power amplifier common to all the transmitting antennas. Furthermore, the considered channel matrix is assumed to be an identity matrix.

The authors of~\cite{ElMoslimany2016} evaluate capacity bounds for $2 \times 2$ multiple input multiple output (MIMO) systems under rectangular peak amplitude constraints and any arbitrary channel matrix. In~\cite{Dytso}, the authors further generalize the investigation to higher dimensional vector Gaussian channels and derive bounds for arbitrary constraint regions. In~\cite{Dytso1shell}, interesting insights on the capacity-achieving input distribution for low signal-to-noise ratio (SNR) levels are presented. Finally, in~\cite{ourITW2020,arxivSP} the present authors derive an upper bound for arbitrary convex constraint regions that, together with the entropy power inequality (EPI) lower bound, provides a vanishing capacity gap at high SNR.

In~\cite[Appendix~F]{li2020capacity}, the authors use a duality-based upper-bounding technique and a suitable auxiliary product output distribution to derive an upper bound that is given by a sum of upper bounds on independent sub-spaces. They derive their upper bound for a system with number of transmitting antennas $N_T$ strictly larger than the number of receiving antennas $N_R$ and for an input constraint region defined as a Cartesian product of $N_R$ one-dimensional sub-regions. 

\subsection*{Contributions}

In this paper, we adapt the result in~\cite[Appendix~F]{li2020capacity} to the case of $\N \times \N$ MIMO systems and generalize their approach to input constraint regions defined as the Cartesian product of an arbitrary number of sub-regions $\K \leq \N$. In addition to the mentioned transmitter configuration using a single power amplifier, another configuration of practical interest is that of employing separate power amplifiers for each transmitting antenna. For this latter case, the resulting constraint region turns out to be a Cartesian product of the constraint imposed by each amplifier, which we refer to as \emph{per-antenna} constraint. 

In this work, we further generalize the constraint region as a Cartesian product of sub-regions lying in sub-spaces of the MIMO system. This generalization can model the transmitter configuration employing multiple power amplifiers, each one feeding a given subset of the transmitting antennas. We propose two upper bounds targeting peak amplitude constraints that can be decomposed into a Cartesian product of sub-regions. The first upper bound that we derive is suitable for high SNR values, and we prove that it converges to the EPI lower bound for increasing SNR. We also define an upper bound suitable for low SNR values. Finally, we apply our bounds to the practical scenario of the per-antenna constraint, which becomes a special case of the considered Cartesian constraint regions.

\subsection*{Paper Organization}

In Sec.~\ref{S:sysmod} we define the channel model, while in Sec.~\ref{S:mainres} we present our main results. We provide high and low SNR regime upper bounds and we investigate their asymptotic behavior. Furthermore, in Sec.~\ref{S:PAnumeric} we specialize the derived upper bounds to the per-antenna constraint an provide numerical results verifying the predicted asymptotic behavior. Finally, Sec.~\ref{S:conclusion} concludes the paper.

\subsection*{Notation}
We use bold letters for vectors ($\xv$) and uppercase letters for random variables ($X$). We represent the $n \times 1$ vector of zeros by $\bm{\mathsf{0}}_n$ and the $n \times n$ identity matrix by $\Iv_n$. We denote by ${\cal CN}(\boldsymbol{\mu},\mathsf{\Sigma})$ a multivariate complex Gaussian distribution and by ${\cal N}(\boldsymbol{\mu},\mathsf{\Sigma})$ a multivariate real Gaussian distribution, both with mean vector $\boldsymbol{\mu}$ and covariance matrix $\mathsf{\Sigma}$. For a given matrix $\Hv$, we define by $\lambda_i\lr{\mathsf{\Hv}}$ the $i$th singular value of $\Hv$. Finally, by $\ball[n]{\R}$ we denote the $n$-dimensional closed ball of radius $\R$ and we define the \mbox{$n$-dimensional} box of sides $\R$ as $\Bocs[n]{\R} \triangleq \{ \xv :  \abs{x_i} \leq \R/2 , \ i=1,\dots,n \}$.

\section{Channel Model} \label{S:sysmod}

Let us consider an $\N \times \N$ real MIMO system with input-output relationship given by
\begin{align} \label{eq:model}
\Yv &=  \Hv \Xv + \sigma_z \Zv,
\end{align}
where $\Yv \in \mathbb{R}^\N$ is the output vector, $\Hv$ is any full rank channel fading matrix, $\Xv \in \sX \subset \mathbb{R}^\N$ is the input vector, with $\sX$ being the input constraint region, and $\Zv \in \mathbb{R}^\N$ is a noise vector such that $\Zv \sim {\cal N}(\bm{\mathsf{0}}_{\N}, \varz \Iv_{\N})$. Let us assume $\Hv$ to be constant over all channel uses and known both at the transmitter and at the receiver.

Throughout this paper we consider input constraint regions that can be decomposed into the Cartesian product of sub-regions. Let us denote by $\K$ the number of sub-regions in $\sX$. We define
\begin{align} \label{eq:cartesianX}
    \sX \triangleq \sX_1 \times \sX_2 \times \dots \times \sX_\K,
\end{align}
where the operator $\times$ denotes the Cartesian product and $\sX_i \subset \mathbb{R}^{\N_i}$ is the $i$th $\N_i$-dimensional sub-region of $\sX$. For convenience in the indexing notation, we also define $\N_0 = 0$. It is worth observing that $\forall i>0, \ \N_i \in \mathbb{N^+}$ and that $\sum_{i=1}^{\K} \N_i = \N$. Let us define the maximum radius of each sub-region as
\begin{align}
\R_i \triangleq \sup_{\xv \in \sX_i} \norm{\xv},
\end{align}
for all $i = 1, \dots, \K$. For the sake of simplicity we will assume $\R_i = \R, \ i = 1,\dots,\K$. Note that, setting all the $\R_i$'s to $\R$ can be done without loss of generality, by scaling the related sub-spaces of $\Hv$ accordingly. 

Since we consider peak amplitude-constrained input distributions, we resort to the following SNR definition
\begin{align}
    \text{SNR} = \frac{\R^2}{\N \varz}.
\end{align}
Finally, we define the channel capacity as
\begin{align} 
    \C &\triangleq \max_{P_\Xv: \: \text{supp}(P_\Xv)\subseteq {\sX}} \mi{\Xv}{\Yv}, 
\end{align}
where $P_\Xv$ is the input distribution law.

\section{Main Results} \label{S:mainres}
In this section we derive two upper bounds. The first upper bound is suitable for the high SNR regime and since in this SNR range the input signal is predominant over the noise, the MIMO capacity can be approximated, broadly speaking, by a sum of capacities, each induced by the sub-spaces where the $\sX_i$'s lie.

Furthermore, we introduce a second upper bound, suitable for the low SNR regime. For this latter SNR range, we assume the Gaussian noise to be the dominant component and, therefore, we upper-bound the capacity by using a Gaussian output distribution.

\subsection{High SNR regime}

To derive an upper bound on the channel capacity, suitable for the high SNR regime, we consider an equivalent output multiplied by the inverse of the channel matrix $\Hv$. Note that, the receiver can compute $\Hv^{-1}$ because the matrix $\Hv$ is full rank and it is known at the receiver. We have
\begin{align} 
 \Hv^{-1} \Yv &= \Hv^{-1} \Hv \cdot \Xv + \Hv^{-1} \Zv \\ 
&= \Xv  + \Hv^{-1} \Zv \\
&= \Xv + \Zv_{\Dv},
\end{align}
where $\Zv_{\Dv} = \Hv^{-1} \Zv$ is the resulting noise vector with $\Zv_{\Dv} \sim \mathcal{N}\lr{\bm{\mathsf{0}}_{\N},\Dv}$ and $\Dv = \varz \Hv^{-1} \Hv^{-T}$.
Let us denote by $d_{k,l}$ the element $(k,l)$ of the matrix $\Dv$ and define the main-diagonal block submatrices $\Dv_i$'s of $\Dv$ as
\begin{align} 
    \Dv_i \triangleq \lrs{d_{k,l}}_{k,l=m_i+1}^{m_i+\N_i}, \quad i = 1,\dots,\K, \label{eq:prinsubm}
\end{align}
where $m_i = \sum_{j=1}^{\N_{i-1}} \N_j$ and $m_1 = 0$. Furthermore, let us denote by $\Xv_i$ the $\N_i \times 1$ vector $\Xv_i = \lr{X_{m_i+1}, X_{m_i+2}, \dots, X_{m_i+\N_i}}^T$ and $\Zv_{\Dv,i}$ analogously.
\begin{theorem} \label{thm:UB}
Given the input constraint region $\sX$ defined in~\eqref{eq:cartesianX}, the channel capacity is upper-bounded by
\begin{align} \label{eq:capacity}
    \C \leq \overline{\C}_1 \triangleq \lr{ \sum_{i=1}^{\K} \C_i } +\frac{1}{2}\log\frac{\prod_{j=1}^\K \det \lr{\Dv_j}}{\det \lr{\Dv}},
\end{align}
where
\begin{align}
\C_i \triangleq \max_{P_{\Xv_i} : \: \Xv_{i} \in \sX_{i}}  \ent{\Xv_i + \Zv_{\Dv,i}} -\ent{ \Zv_{\Dv,i}}
\end{align}
and $\Zv_{\Dv,i} \sim \mathcal{N} \lr{\bm{\mathsf{0}}_{\N_i},\Dv_i}$.
\end{theorem}
\begin{proof}
\begin{align} 
\C &= \max_{P_\Xv : \: \Xv \in \sX} \mi{\Xv}{\Hv\Xv + \Zv}\\
&= \max_{P_\Xv : \: \Xv \in \sX} \mi{\Xv}{\Xv + \Zv_{\Dv}}\\
&= \max_{P_\Xv : \: \Xv \in \sX} \ent{\Xv + \Zv_{\Dv}} - \ent{\Zv_{\Dv}}\\
\begin{split}
&\stackrel{(a)}{\leq} \lr{\max_{P_\Xv : \: \Xv \in \sX} \sum_{i=1}^{\K} \ent{\Xv_i + \Zv_{\Dv,i}}} +\log\frac{1}{\det \lr{\Dv}} \\
&\hphantom{\stackrel{(a)}{\leq} \ }- \frac{\N}{2} \log \lr{2 \pi e} 
\end{split} \\
\begin{split}
&= \lr{ \sum_{i=1}^{\K} \max_{P_{\Xv_i} : \: \Xv_{i} \in \sX_{i}}  \ent{\Xv_i + \Zv_{\Dv,i}}} +\frac{1}{2} \log\frac{1}{\det \lr{\Dv}} \\
&\hphantom{= \ }- \frac{\N}{2}\log \lr{2 \pi e} + \frac{1}{2}\log \frac{\prod_{j=1}^\K \det \lr{\Dv_j}}{\prod_{k=1}^\K \det \lr{\Dv_k}}
\end{split} \label{eq:add&sub} \\
\begin{split}
&= \lr{ \sum_{i=1}^{\K} \max_{P_{\Xv_i} : \: \Xv_{i} \in \sX_{i}}  \ent{\Xv_i + \Zv_{\Dv,i}} - \ent{\Zv_{\Dv,i}}}  \\
&\hphantom{= \ }+\frac{1}{2}\log\frac{\prod_{j=1}^\K \det \lr{\Dv_j}}{\det \lr{\Dv}},
\end{split} \label{eq:Ciproof}
\end{align}
where $(a)$ holds because of the sub-additivity of the differential entropy and $\Zv_{\Dv,i}$ is obtained by marginalizing $\Zv_{\Dv}$ on the $i$th sub-space. Note that, since $\Zv_{\Dv}$ is a multivariate Gaussian with zero mean and covariance matrix $\Dv$, it holds $\Zv_{\Dv,i} \sim \mathcal{N} \lr{\bm{\mathsf{0}}_{\N_i},\Dv_i}$. In~\eqref{eq:add&sub}, we add and subtract $\frac{1}{2} \log \prod_i \det \lr{\Dv_i}$, to finally get the term $\ent{\Zv_{\Dv,i}}$ in~\eqref{eq:Ciproof}.
\end{proof}
\begin{remark}
To obtain a numerical result, each $\C_i$ can be further upper-bounded with a suitable technique, like those presented in~\cite{arxivSP,thangaraj2017capacity,Dytso}.
\end{remark}
\begin{remark} Since $\Dv$ is positive-semidefinite, by Fischer’s inequality~\cite{horn2012matrix}, we have that $\det \lr{\Dv} \leq \prod_{j=1}^{\K} \det \lr{\Dv_j}$.
Therefore, it holds
\begin{align}
    \log\frac{\prod_{j=1}^\K \det \lr{\Dv_j}}{\det \lr{\Dv}} \ge 0.
\end{align}
\end{remark}
\begin{remark}
Intuitively, the logarithmic term in~\eqref{eq:capacity} accounts for the inaccuracy introduced by considering the noise vector $\Zv_{\Dv}$ to be independent on each of the $\K$ sub-spaces. Indeed, whenever $\Hv$ is diagonal we have that $\det\lr{\Dv} = \prod_{i=1}^{\K} \det \lr{\Dv_i}$, then $\log\frac{\det \lr{\Dv}}{\prod_j \det \lr{\Dv_j}}$ goes to zero and inequality~\eqref{eq:capacity} becomes an equality.
\end{remark}

Let us introduce the EPI lower bound~\cite{Dytso} for the channel in~\eqref{eq:model} as
\begin{align}
    \C \geq \underline{\C} &\triangleq \frac{\N}{2} \log \lr{ 1 + \frac{ \lr{\Vol[\N]{\Hv \sX}}^{\frac{2}{\N}} }{2 \pi e \varz}}. \label{eq:EPI}
\end{align}
In the following lemma, we show that the capacity gap between the EPI lower bound and the upper bound in Theorem~\ref{thm:UB} is vanishing when the SNR tends to infinity.
\begin{lemma} \label{lem:UBgap}
When $\varz \to 0$, we have
\begin{align}
    \lim_{\varz \to 0} \overline{\C}_1 - \underline{\C} = 0.
\end{align}
\end{lemma}
\begin{proof}
Let us consider the mutual information for the $i$th subchannel $\mi{\Xv_i}{\Xv_i + \Zv_{\Dv,i}} = \ent{\Xv_i + \Zv_{\Dv,i}} - \ent{\Zv_{\Dv,i}}$. Let us denote by $\Mv_i$ the $\N_i \times \N_i$ matrix such that $\Dv_i = \varz \Mv_i^{-1}\Mv_i^{-T}$. We can derive such matrix $\Mv_i$ because $\Dv_i$ is a covariance matrix and therefore it is positive-semidefinite. We have that
\begin{align}
    \mi{\Xv_i}{\Xv_i + \Zv_{\Dv,i}} &= \mi{\Xv_i}{\Mv_i\Xv_i + \Zv_{i}},
\end{align}
where $\Zv_{i}\sim \mathcal{N}(\bm{\mathsf{0}}_{\N_i}, \varz \Iv_{\N_i})$. Then, we have
\begin{align}
    \lim_{\varz \to 0} \sum_{i=1}^{\K} \C_i  &=  \sum_{i=1}^{\K} \max_{P_{\Xv_i} : \: \Xv_{i} \in \sX_{i}} \ent{\Mv_i\Xv_i} - \lim_{\varz \to 0} \ent{\Zv_i}\\
    & =  \lr{\sum_{i=1}^{\K} \log \Vol[\N_i]{\Mv_i \sX_i}} - \lim_{\varz \to 0}\ent{\Zv} \label{eq:uniformmaxdist} \\
    & =  \lr{\sum_{i=1}^{\K} \log \det \lr{\Mv_i}\Vol[\N_i]{ \sX_i}} -\lim_{\varz \to 0} \ent{\Zv},
\end{align}
where~\eqref{eq:uniformmaxdist} holds because $\ent{\Mv_i\Xv_i}$ is maximized by the uniform distribution over $\sX_i$. Notice also that
\begin{align}
    &\lim_{\varz \to 0} \frac{1}{2}\log\frac{\prod_{i=1}^\K \det \lr{\Dv_i}}{\det \lr{\Dv}} \\
    & \quad = \lim_{\varz \to 0} \frac{1}{2}\log\frac{\prod_{i=1}^\K \det \lr{\varz \Mv_i^{-1}\Mv_i^{-T}}}{\det \lr{\varz \Hv^{-1}\Hv^{-T}}} \\
    & \quad = \log\det\lr{\Hv} - \sum_{i=1}^{\K} \log \det \lr{ \Mv_i },
\end{align}
where we used the fact that $\det(\Hv^{-1}\Hv^{-T}) = (\det(\Hv^{-1}))^2 = 1/(\det(\Hv))^2$ and similarly for the $\Mv_i$'s. 

Furthermore, for the lower bound it holds that
\begin{align}
    \lim_{\varz \to 0} \underline{\C} &= \lim_{\varz \to 0} \frac{\N}{2}\log\lr{\frac{ \lr{\Vol[\N]{\Hv \sX}}^{\frac{2}{\N}} }{2 \pi e \varz}} \\
    &= \lim_{\varz \to 0} \log\lr{\Vol[\N]{\Hv \sX}} - \ent{\Zv} \\
    \begin{split}
    &= \lim_{\varz \to 0} \log \det \lr{\Hv} + \log\lr{\prod_{i=1}^{\K}\Vol[\N_i]{ \sX_i}} \\
    & \hphantom{= \lim_{\varz \to 0} \ } - \ent{\Zv}.
    \end{split}
\end{align}
Notice that, since $\sX$ is defined by a Cartesian product, it holds that $\Vol[\N]{\sX} = \prod_i \Vol[\N_i]{\sX_i}$.

Finally, by putting everything together we get
\begin{align}
    \lim_{\varz \to 0} \overline{\C}_1 - \underline{\C} = 0.
\end{align}
\end{proof}
\begin{remark}
Whenever the constraint sub-regions $\sX_i$'s are convex, we can always derive an upper bound on the $\C_i$'s by applying the sphere packing upper bound in~\cite{ourITW2020}. Since in the mentioned paper we proved that the upper bound asymptotically converges to $\frac{\N}{2} \log \lr{\frac{\lr{\Vol[\N]{\Hv \sX}}^{2/ \N}}{2 \pi e \varz}}$ for large SNR, we have that an upper bound satisfying Lemma~\ref{lem:UBgap} can always be evaluated for any full rank $\Hv$ and any convex region $\sX$. 
\end{remark}

\subsection{Low SNR regime}

At low SNR, i.e., when the Gaussian noise is dominant, the upper bound in Theorem~\ref{thm:UB} is loose. Intuitively, as $\R$ goes to zero, $\sX$ becomes smaller and the output distribution becomes closer to a Gaussian. Therefore, we can derive an upper bound tighter than~\eqref{eq:capacity} by using a Gaussian maximum-entropy argument.

Let us consider the singular value decomposition of $\Hv$, i.e., $\Hv = \Uv \Lv \Vv^T$. Given~\eqref{eq:model}, we can consider the equivalent model
\begin{align}
   \Lv^{-1} \Bar{\Yv} &= \Bar{\Xv} + \Lv^{-1} \Bar{\Zv} \\
   &= \Bar{\Xv} + \Bar{\Zv}_{\Bar{\Dv}}, \label{eq:equivmodel}
\end{align}
where $\Bar{\Yv} = \Uv^{-1} \Yv$, the input is $\Bar{\Xv} = \Vv^T \Xv$, and the noise vector is $\Bar{\Zv}_{\Bar{\Dv}} =\Lv^{-1}\Bar{\Zv} = \Lv^{-1} \Uv^{-1} \Zv$. Notice that since $\Zv$ has a rotationally symmetric distribution, we still have $\Bar{\Zv} \sim \mathcal{N}(\bm{\mathsf{0}}_{\N}, \varz \Iv_{\N})$ and $\Bar{\Zv}_{\Bar{\Dv}}  \sim {\cal N} \lr{\bm{\mathsf{0}}_{\N} , \Bar{\Dv} }$ with $\Bar{\Dv} = \varz \Lv^{-1}\Lv^{-T}$.
\begin{theorem} \label{thm:lowUB}
Given the input constraint region $\sX$ defined in~\eqref{eq:cartesianX}, the channel capacity is upper-bounded by
\begin{align}
    \C \leq \overline{\C}_2 \triangleq \lr{ \sum_{i=1}^{\N} \frac{1}{2} \log \lr{ \Pj_i + \lambda_i\lr{\Dv } }  } -\frac{1}{2}\log \det \lr{ \Dv }, \label{eq:lowUB}
\end{align}
where $\Pj_i$ is the power allocation given by the water-filling algorithm, for a total available average power $\R^2\K$ and $\N$ parallel channels with noise variances $\lambda_i\lr{\Dv}$'s.
\end{theorem}
\begin{proof}

Since $\sX$ is a Cartesian product of $\K$ sub-regions, each one contained in a ball of radius $\R$, we have that $\sup_{\xv \in \sX}\lrc{ \norm{ \xv } } = \R\sqrt{\K}$. Given the constraint imposed by $\sX$, the looser constraint $\expect{\Xv^T \Xv} \leq \R^2 \K$ is always satisfied. We have
\begin{align} 
\C &= \max_{P_\Xv : \: \Xv \in \sX} \mi{\Xv}{\Hv\Xv + \Zv}\\
& = \max_{\substack{P_\Xv : \: \Xv \in \sX,\\ \ \expect{\Xv^T \Xv}  \leq \R^2 \K}}\mi{\Xv}{\Hv\Xv + \Zv}  \\
&  \leq \max_{P_\Xv : \: \expect{\Xv^T \Xv}  \leq \R^2 \K} \mi{\Xv}{\Hv\Xv + \Zv}  \label{eq:looserconstr}\\
& = \max_{P_{\Bar{\Xv}} : \: \expect{\Bar{\Xv}^T \Bar{\Xv}}  \leq \R^2 \K} \mi{\Bar{\Xv}}{\Bar{\Xv} + \Bar{\Zv}_{\Bar{\Dv}}} \label{eq:equivsys} \\
& = \max_{P_{\Bar{\Xv}} : \: \expect{\Bar{\Xv}^T \Bar{\Xv}} \leq \R^2 \K} \ent{\Lv^{-1} \Bar{\Yv}} - \ent{\Bar{\Zv}_{\Bar{\Dv}}}\\
& \leq \max_{P_{\Bar{\Xv}} : \: \expect{\Bar{\Xv}^T \Bar{\Xv}} \leq \R^2 \K} \ent{\widetilde{\Yv}} - \ent{\Bar{\Zv}_{\Bar{\Dv}}} \label{eq:gme} \\
&  \leq \max_{P_{\Bar{\Xv}} : \: \expect{\Bar{\Xv}^T \Bar{\Xv}} \leq \R^2 \K}  \sum_{i=1}^{\N} \frac{1}{2} \log \lr{ 2 \pi e \lr{ \expect{\abs{\Bar{X}_i}^2} + \lambda_i\lr{\Bar{\Dv}}  }} \nonumber \\
& \hphantom{\leq \max_{P_{\Bar{\Xv}} : \: \expect{\Bar{\Xv}^T \Bar{\Xv}} \leq \R^2 \K}} - \frac{1}{2} \log \det \lr{2 \pi e \Bar{\Dv}}
\label{eq:xuncorr} \\
& = \lr{ \sum_{i=1}^{\N} \frac{1}{2} \log \lr{  \Pj_i +  \lambda_i\lr{\Dv} } } - \frac{1}{2} \log \det \lr{\Dv} , \label{eq:waterfill}
\end{align} 
where the upper bound in~\eqref{eq:looserconstr} holds because we removed the constraint imposed by $\sX$, in~\eqref{eq:equivsys} we used the equivalent model defined in~\eqref{eq:equivmodel}. Since $\Vv^T$ is a unitary matrix, we have that $\expect{\Xv^T \Xv} = \expect{\Bar{\Xv}^T \Bar{\Xv}}$. Let us define the normally distributed vector $\widetilde{\Yv} \sim {\cal N} \lr{\bm{\mathsf{0}}_{\N} , \mathsf{\Sigma} }$, with $\mathsf{\Sigma} = \expect{\Bar{\Xv} \Bar{\Xv}^T} + \Bar{\Dv}$. For the upper bound in~\eqref{eq:gme} we used a Gaussian maximum-entropy bound $\ent{\Lv^{-1}\Bar{\Yv}} \leq \h (\widetilde{\Yv})$, and in~\eqref{eq:xuncorr} we used $\h (\widetilde{\Yv}) \leq \sum_i \h (\widetilde{Y}_i)$. Finally, to obtain~\eqref{eq:waterfill} we notice that $\lambda_i\lr{\Bar{\Dv}} = \lambda_i\lr{\Dv}$ for any $i$ and we apply the water-filling algorithm.
\end{proof}

The following trivial lemma shows that the upper bound $\overline{\C}_2$ is suitable for the low SNR regime.
\begin{lemma} \label{lem:lowUBgap}
The capacity upper bound $\overline{\C}_2$ tends to zero for $\varz \to \infty$
\begin{align}
    \lim_{\varz \to \infty} \overline{\C}_2 = 0.
\end{align}
\end{lemma}
\begin{proof}
In Theorem~\ref{thm:lowUB}, when $\varz \to \infty$ the $\Pj_i$'s tend to be negligible compared to the $\lambda_i\lr{\Dv}$, which are proportional to $\varz$. Therefore, we have that
\begin{align}
    \lim_{\varz \to \infty} \overline{\C}_2 =  \lr{ \sum_{i=1}^{\N} \frac{1}{2} \log \lr{ \lambda_i\lr{\Dv} }} - \frac{1}{2} \log \det \lr{\Dv} = 0 .
\end{align}
\end{proof}

\section{Per-Antenna Constraint} \label{S:PAnumeric}

The proposed upper bounds can be applied to a common and practical constraint, namely the \emph{per-antenna} constraint. A transmitter configuration of practical interest in MIMO systems is that of a single power amplifier for each transmitting antenna. We model the transmitted signal on each antenna as a complex signal. Let us consider a MIMO system with $\N/2$ complex dimensions
\begin{align} \label{eq:complexmodel}
\Yv' &=  \Hv' \Xv' + \Zv',
\end{align}
where $\Yv' \in \mathbb{C}^{\N/2}$ is the output vector, $\Hv'$ is any full rank channel fading matrix, $\Xv' \in \sX' = \Bocs[\N/2]{2\R} \subset \mathbb{C}^{\N/2}$ is the input vector, with $\sX'$ being the input constraint region, and $\Zv' \in \mathbb{C}^{\N/2}$ is a noise vector such that $\Zv' \sim {\cal CN}(\bm{\mathsf{0}}_{\N}, 2\varz\Iv_{\N})$.
Note that, we can still refer to the model in~\eqref{eq:model} simply by vectorizing the system in~\eqref{eq:complexmodel}: We need to define $\Hv = \text{Re}\{ \Hv' \} \otimes \Iv_{2} + \text{Im}\{ \Hv' \} \otimes \begin{bsmallmatrix}0 & -1\\1 & \ 0\end{bsmallmatrix}$, where the operator~$\otimes$ is the Kronecker product, while the output vector is such that $\Yv = [\text{Re}(Y'_1),\text{Im}(Y'_1),\dots,\text{Re}(Y'_\N),\text{Im}(Y'_\N)]^T$, and analogously for $\Xv$ and $\Zv$. Finally, notice that for any $i=1,\dots,\N/2$ the constraint $\abs{X'_i} \leq \R$ is equivalent to
\begin{align}
    \Xv_i \in \sX_i = \ball[2]{\R}, \ i=1,\dots,\N/2,
\end{align}
where $\Xv_i = \lr{\begin{smallmatrix} \text{Re}(X'_i) \\ \text{Im}(X'_i) \end{smallmatrix}}$. While the upper bound in Theorem~\ref{thm:lowUB} can be applied directly, the upper bound $\overline{\C}_1$ of Theorem~\ref{thm:UB} has to be specialized for the per-antenna case. Let us consider the following equivalent expression of $\C_i$ defined in the proof of Lemma~\ref{lem:UBgap}
\begin{align}
    \C_i = \max_{P_{\Xv_i} : \: \Xv_{i} \in \sX_{i}} \ent{\Mv_i\Xv_i} - \ent{\Zv_i}.
\end{align}
Since $\Hv$ is obtained by vectorizing $\Hv'$, the singular values of $\Hv$ are equal $2$-by-$2$, i.e., $\lambda_{2i} \lr{\Hv}=\lambda_{2i-1} \lr{\Hv}, \ i=1,\dots,\N/2$. The same is true for the singular values of $\Dv$, $\Dv_i$'s, and $\Mv_i$'s. To simplify the notation, we define
\begin{align}
    \lambda \lr{\Mv_i} \triangleq \lambda_1 \lr{\Mv_i} = \lambda_2 \lr{\Mv_i}, \quad i=1,\dots,\N/2.
\end{align}
In the per-antenna case, suitable upper bounds for each $\C_i$ are defined in~\cite{thangaraj2017capacity}. The McKellips-Type upper bound, derived in~\mbox{\cite[Eq.~(32)]{thangaraj2017capacity}}, gives the following simple closed form expression
\begin{align}
    \C_i \leq \log \lr{ 1 + \sqrt{\frac{\pi}{2}} \frac{\lambda \lr{\Mv_i} \R}{\sigma_z} + \frac{\lr{\lambda \lr{\Mv_i}\R}^2}{2 e\varz}}. \label{eq:McKUB}
\end{align}
Therefore, we have
\begin{align}
    \overline{\C}_1 &\leq \overline{\C}_{\text{PA},1} \\ 
    \begin{split}
    &= \lr{ \sum_{i=1}^{\N/2} \log \lr{ 1 + \sqrt{\frac{\pi}{2}} \frac{\lambda \lr{\Mv_i} \R}{\sigma_z} + \frac{\lr{\lambda \lr{\Mv_i}\R}^2}{2 e\varz}} } \\ &\hphantom{\leq} \ +\frac{1}{2}\log\frac{\prod_{j=1}^\K \det \lr{\Dv_j}}{\det \lr{\Dv}}.
    \end{split} \label{eq:PA_McKUB}
\end{align}
Furthermore, the authors of~\cite{thangaraj2017capacity} derive an additional upper bound, tighter than~\eqref{eq:McKUB}, that however has to be computed via a numerical optimization. For a given a given $\C_i$, let us denote by $\overline{\C}_{\text{Ref},i} \geq \C_i$ this \emph{refined} upper bound~\cite[Eq.~(82)]{thangaraj2017capacity}. Then, by plugging the $\overline{\C}_{\text{Ref},i}$'s into Theorem~\ref{thm:UB}, we define $\overline{\C}_{\text{PA},2}$ as follows
\begin{align}
    \C \leq \overline{\C}_1 \leq \overline{\C}_{\text{PA},2} \triangleq \lr{\sum_{i=1}^{\K} \overline{\C}_{\text{Ref},i} }+\frac{1}{2}\log\frac{\prod_{j=1}^\K \det \lr{\Dv_j}}{\det \lr{\Dv}}. \label{eq:PA_RefUB}
\end{align}

\subsection{Numerical Results}

For the per-antenna case, let us now evaluate numerically $\underline{\C}$, $\overline{\C}_2$, and both the specialized versions of $\overline{\C}_1$. We evaluate the bounds for a random realization of $\Hv$.
\begin{figure}[t]
    \centering
    \input{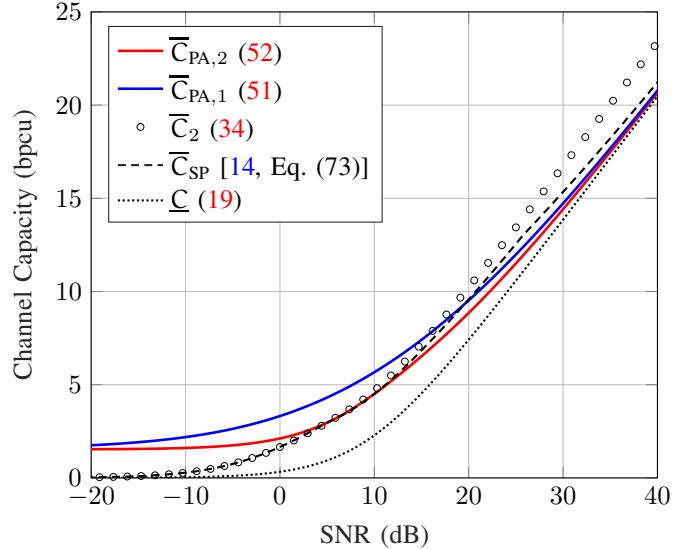}
    \caption{Capacity bounds in bit per channel use (bpcu) versus SNR, for $\N=4$, $\lambda \lr{\Mv_1} = 0.52$, and $\lambda \lr{\Mv_2} = 0.37$.}
    \label{fig:capacityPA}
\end{figure}%
If we consider the compound upper bound given by $\min\lr{\overline{\C}_2, \overline{\C}_{\text{PA},2} }$ we see that, as predicted by Lemma~\ref{lem:UBgap} and Lemma~\ref{lem:lowUBgap}, the capacity gap between upper and lower bounds is indeed vanishing both at high SNR, thanks to $\overline{\C}_{\text{PA},2}$, and at low SNR, thanks to $\overline{\C}_2$. Moreover, we compare the proposed bounds to the previous best in the existing literature, which we proposed in~\cite{arxivSP}. Specifically, let us denote by $\overline{\C}_{\text{SP}}$ the upper bound based on a sphere packing argument~\mbox{\cite[Eq.~(73)]{arxivSP}}. The sphere packing has also the properties of being vanishing at both low and high SNR, but as seen in Fig.~\ref{fig:capacityPA}, the upper bounds $\overline{\C}_{\text{PA},1}$ and $\overline{\C}_{\text{PA},2}$, derived from of Theorem~\ref{thm:UB}, can improve the tightness of the capacity gap also at finite SNR levels of practical interest.

\section{Conclusion} \label{S:conclusion}

We have derived two upper bounds on the channel capacity of peak amplitude-constrained vector Gaussian channels affected by fading. We considered constraint regions that can be decomposed into a Cartesian product, reflecting the fact that each power amplifier feeds a subset of the transmitting antennas. We have proved that the first upper bound, suitable for the high signal-to-noise (SNR) regime, has vanishing capacity gap when compared to the entropy power inequality lower bound. The second proposed upper bound is suitable for the low SNR regime. Finally, for a transmitter that employs separate power amplifiers for each antenna, we have shown an example where the proposed upper bounds are tighter than the best known upper bounds at any SNR.

\bibliographystyle{IEEEtran}
\bibliography{bibliofile}

\begin{thebibliography}{10}
\providecommand{\url}[1]{#1}
\csname url@samestyle\endcsname
\providecommand{\newblock}{\relax}
\providecommand{\bibinfo}[2]{#2}
\providecommand{\BIBentrySTDinterwordspacing}{\spaceskip=0pt\relax}
\providecommand{\BIBentryALTinterwordstretchfactor}{4}
\providecommand{\BIBentryALTinterwordspacing}{\spaceskip=\fontdimen2\font plus
\BIBentryALTinterwordstretchfactor\fontdimen3\font minus
  \fontdimen4\font\relax}
\providecommand{\BIBforeignlanguage}[2]{{%
\expandafter\ifx\csname l@#1\endcsname\relax
\typeout{** WARNING: IEEEtran.bst: No hyphenation pattern has been}%
\typeout{** loaded for the language `#1'. Using the pattern for}%
\typeout{** the default language instead.}%
\else
\language=\csname l@#1\endcsname
\fi
#2}}
\providecommand{\BIBdecl}{\relax}
\BIBdecl

\bibitem{Smith}
J.~G. Smith, ``The information capacity of amplitude- and variance-constrained
  scalar {G}aussian channels,'' \emph{Information and Control}, vol.~18, no.~3,
  pp. 203--219, April 1971.

\bibitem{Shamai}
S.~{Shamai} and I.~{Bar-David}, ``The capacity of average and
  peak-power-limited quadrature {G}aussian channels,'' \emph{IEEE Transactions
  on Information Theory}, vol.~41, no.~4, pp. 1060--1071, July 1995.

\bibitem{Rassouli}
B.~{Rassouli} and B.~{Clerckx}, ``On the capacity of vector {G}aussian channels
  with bounded inputs,'' \emph{IEEE Transactions on Information Theory},
  vol.~62, no.~12, pp. 6884--6903, December 2016.

\bibitem{Tchamkerten2004}
A.~Tchamkerten, ``On the discreteness of capacity-achieving distributions,''
  \emph{IEEE Transactions on Information Theory}, vol.~50, no.~11, pp.
  2773--2778, Nov. 2004.

\bibitem{Chan2005}
T.~Chan, S.~Hranilovic, and F.~Kschischang, ``Capacity-achieving probability
  measure for conditionally {G}aussian channels with bounded inputs,''
  \emph{IEEE Transactions on Information Theory}, vol.~51, no.~6, pp.
  2073--2088, Jun. 2005.

\bibitem{Mamandipoor2014}
B.~Mamandipoor, K.~Moshksar, and A.~K. Khandani, ``Capacity-achieving
  distributions in {G}aussian multiple access channel with peak power
  constraints,'' \emph{IEEE Transactions on Information Theory}, vol.~60,
  no.~10, pp. 6080--6092, Oct. 2014.

\bibitem{McKellips}
A.~L. {McKellips}, ``Simple tight bounds on capacity for the peak-limited
  discrete-time channel,'' in \emph{International Symposium on Information
  Theory, 2004. ISIT 2004. Proceedings.}, June 2004, pp. 348--348.

\bibitem{thangaraj2017capacity}
A.~{Thangaraj}, G.~{Kramer}, and G.~{Böcherer}, ``Capacity bounds for
  discrete-time, amplitude-constrained, additive white {G}aussian noise
  channels,'' \emph{IEEE Transactions on Information Theory}, vol.~63, no.~7,
  pp. 4172--4182, Jul. 2017.

\bibitem{ourISIT2021}
A.~Favano, M.~Ferrari, M.~Magarini, and L.~Barletta, ``The capacity of the
  amplitude-constrained vector {G}aussian channel,'' in \emph{2021 IEEE
  International Symposium on Information Theory (ISIT)}, Jul. 2021, pp.
  426--431.

\bibitem{ElMoslimany2016}
A.~ElMoslimany and T.~M. Duman, ``On the capacity of multiple-antenna systems
  and parallel {G}aussian channels with amplitude-limited inputs,'' \emph{IEEE
  Transactions on Communications}, vol.~64, no.~7, pp. 2888--2899, Jul. 2016.

\bibitem{Dytso}
A.~{Dytso}, M.~{Goldenbaum}, S.~{Shamai}, and H.~V. {Poor}, ``Upper and lower
  bounds on the capacity of amplitude-constrained {MIMO} channels,'' in
  \emph{GLOBECOM 2017 - 2017 IEEE Global Communications Conference}, December
  2017, pp. 1--6.

\bibitem{Dytso1shell}
A.~{Dytso}, M.~{Al}, H.~V. {Poor}, and S.~{Shamai Shitz}, ``On the capacity of
  the peak power constrained vector {G}aussian channel: An estimation theoretic
  perspective,'' \emph{IEEE Transactions on Information Theory}, vol.~65,
  no.~6, pp. 3907--3921, January 2019.

\bibitem{ourITW2020}
A.~Favano, M.~Ferrari, M.~Magarini, and L.~Barletta, ``A sphere packing bound
  for {AWGN} {MIMO} fading channels under peak amplitude constraints,'' in
  \emph{2020 IEEE Information Theory Workshop (ITW)}, 2021, pp. 1--5.

\bibitem{arxivSP}
------, ``A sphere packing bound for vector {G}aussian fading channels under
  peak amplitude constraints,'' \emph{arXiv preprint arXiv:2111.13179}, Nov
  2021.

\bibitem{li2020capacity}
L.~Li, S.~M. Moser, L.~Wang, and M.~Wigger, ``On the capacity of {MIMO} optical
  wireless channels,'' \emph{IEEE Transactions on Information Theory}, vol.~66,
  no.~9, pp. 5660--5682, 2020.

\bibitem{horn2012matrix}
R.~A. Horn and C.~R. Johnson, \emph{Matrix Analysis}.\hskip 1em plus 0.5em
  minus 0.4em\relax Cambridge university press, 2012.

\end{thebibliography}

\end{document}